\documentclass[10pt,onecolumn]{IEEEtran}
\usepackage{amsfonts}
\usepackage{indentfirst}
\usepackage{amsmath}
\usepackage{amsthm}
\usepackage{amssymb}
\usepackage{color}
\usepackage{graphicx}
\usepackage[para]{threeparttable}
\usepackage{booktabs}
\usepackage{caption}
\usepackage{makecell}
\usepackage{diagbox}
\usepackage{multirow}
\usepackage{epstopdf}

\theoremstyle{remark} 

\newtheorem{Definition}{\textit{Definition}}
\newtheorem{Theorem}{\textit{Theorem}}
\newtheorem{Lemma}{\textit{Lemma}}
\newtheorem{Remark}{\textit{Remark}}
\newtheorem{Example}{\textit{Example}}

\allowdisplaybreaks[4]
\begin{document}
\title{New Construction of Z-Complementary Code Sets and Mutually Orthogonal Complementary Sequence Sets}
\author{Bingsheng Shen, Yang Yang, Zhengchun Zhou}
\maketitle
\begin{abstract}
  Due to the zero nontrivial aperiodic correlation of complete complementary code (CCCs), it is used in asynchronous multi carrier code division multiple access (MC-CDMA) communication to provide zero interference performance. However, there is no CCC for some lengths. In this case, if the system pays more attention to the correlation, it can use mutually orthogonal complementary sequence sets (MOCSSs), and if the system pays more attention to the set size, it can use Z-complementary code sets (ZCCSs). In this paper, we propose a direct construction of $q$-ary $(q^{v+1},q,q^m,q^{m-v})$-ZCCS. In addition, based on the known CCCs, we propose a new construction of MOCSSs and CCCs respectively.
\end{abstract}


\section{Introduction}
Golay proposed a pair of sequences, which the sum of their aperiodic auto-correlation function (AACF) is zero everywhere except at the zero-shift position, known as Golay complementary pair \cite{Golay-1961}. GCPs have found numerous applications in modern day communication systems, radar, image processing \cite{Davis-1999}--\cite{Pezeshki-2008}. For the specified alphabet set, the length of GCP is very limited. For example, the length of known binary GCPs all take the form $2^a10^b26^c$, where $a,~b,~c$ are non-negative integers.
Therefore, Tseng and Liu extended the idea of GCP to complementary sequence sets (CSSs), each set contains multiple sequences, and the sum of their AACF is also zero for all nonzero time shifts \cite{Tseng-1972}. In addition, they also proposed mutually orthogonal complementary sequence sets (MOCSSs). A $(M,N,L)$-MOCSS is a family of $M$  $(N,L)$-sequence sets, where $M$ denotes the set size (i.e., the number of users), $N$ denotes the flock size (i.e., the number of sub-carriers) and $L$ denotes the sequence length \cite{Tseng-1972}.
In 1988, Suehiro and Hatori given an upper bound of the set size of the $(M,N,L)$-MOCSS, i.e., $M\leq N$ \cite{Suehiro-1988}. When the equal sign holds, MOCSS is called a complete complementary code (CCC). Owing to their ideal correlation properties, MOCSSs or CCCs have found applications in synthetic aperture imaging systems \cite{Trots-2015}, OFDM-CDMA systems \cite{Zhang-2014}, especially the application multi-carrier code division multiple access (MC-CDMA) systems \cite{Tseng-2000}--\cite{Liu-2014}, because zero multi-path interference (MPI) and zero multi-user interference (MUI) performance can be achieved.
For more research on CCCs and MOCSSs, please refer to \cite{Liu-2014}--\cite{Shen-Yang-2021}.

A major disadvantage of CCC is that the number of supported users is limited by the number of row sequences in each complementary matrix. This problem can be solved by Z-complementary code sets (ZCCS), which there is a zero correlation zone (ZCZ) in the aperiodic cross- and auto-correlation, and set size is several times that of CCC \cite{Liu-Guan-2011}. A $(M,N,L,Z)$-ZCCS, the upper bound of its set size is given by \cite{Feng-2013}, i.e., $M\leq N\lfloor L/Z\rfloor$. When the equal sign holds, ZCCS is called optimal. If all the received multi-user signals are roughly synchronized within the ZCZ width, the ZCZ property of ZCCS can alleviate multiple-access interference. In addition to their applications in MC-CDMA, ZCCS have also been employed as optimal training sequences in multiple-input multiple-output (MIMO) communications \cite{Wang-Gao-2007,Yuan-Tu-2008}.

In terms of the construction of ZCCS, it is currently divided into the following categories: The first is the direct construction based on generalized Boolean functions (GBFs), and the parameters of this kind of construction are all related to two.
Such as, Wu \emph{et. al.} proposed a construction of $(2^{k+v},2^k,2^m,2^{m-v})$-ZCCS based on second order GBFs, where $v\leq m$, $k\leq m-v$ \cite{Wu-Chen-2018}. Sarkar \emph{et. al.} proposed a construction of $(2^{k+p+1},2^{k+1},2^m,2^{m-p})$-ZCCS based on second order GBFs in 2019, where $k+p\leq m$ \cite{Sarkar-Majhi-2019}. Sarkar explained their construction from the perspective of graph theory of GBF, so the number of ZCCS constructed is more than \cite{Wu-Chen-2018}. Later on, Sarkar \emph{et. al.} present a construction of $(2^{n+p},2^{n},2^m,2^{m-p})$-ZCCS based on higher order $(\geq2)$ GBFs in 2021, where $p\leq m$ \cite{Sarkar-Majhi-2021}. In addition, Sarkar \emph{et. al.} also proposed a construction of ZCCS with non-power-of-two lengths based on GBFs. Wu \emph{et. al.} also put forward a construction of ZCCS with non-power-of-two lengths, which is not optimal, but its length is very flexible \cite{Wu-Chen-2021}.
The second is based on the construction of the Z-paraunitary matrices (ZPU) and unitary matrices \cite{Das-Parampalli-2020,Chen-Li-2020}, this construction depends on the existence of the ZPU and unitary matrix.
The third is based on GCP and Z-complementary pair (ZCP), which can be referred to \cite{Li-Xu-2015},\cite{Adhikary-Majh-2019}.

In this paper, we first propose a construction of $q$-ary $(q^{v+1},q,q^m,q^{m-v})$-ZCCS based on extended Boolean function (EBF).
EBF is a generalization of Boolean function, defined as a mapping from $\mathbb{Z}_q^m$ to $\mathbb{Z}_q$. As far as we know, the parameters of the known direct constructions based on GBFs are always related to $2$.
However, due to the arbitrariness of $q$, the construction of this paper breaks through this limitation.
In addition, we generalize the result of literature \cite{Shen-Yang-2021} so that it can construct $(M_1M_2/2,2L_1L_2)$-CCC based on $(M_1,L_1)$-CCC and $(M_2,L_2)$-CCC.
Finally, we construct a $(M,2M,L_1+L_2)$-MOCSS from $(M,L_1)$-CCC and $(M,L_2)$-CCC.
Although this construction is simple, it can summarize all current constructions in terms of sequence length.

The rest of this paper is organized as follows. In Section II, we review some definitions. We will provide a construction of ZCCS in Section III, and propose a novel construction of CCC and MOCSS in Section VI. In Section V, we will compare our results with the existing works, and conclude this paper.

\section{Preliminary}
Before we begin, let us define some notations, which will be used throughout this paper.
\begin{itemize}
	\item $\xi=e^{\sqrt{-1}\frac{2\pi}{q}}$ denotes the $q$-th root of unity;
	\item $\mathbb{Z}_q=\{0,1,\cdots,q-1\}$ is an additive group of order $q$;
    \item Bold small letter $\mathbf{a}$ denotes a sequence of length $L$, i.e., $\mathbf{a}=(a_0,a_1,\dots,a_{L-1})$;
    \item $\mathcal{I}(\mathbf{a},\mathbf{b})=(a_0,b_0,a_1,b_1,\cdots,a_ {L-1},b_{L-1})$ denotes the interleaved sequence of $\mathbf{a}$ and $\mathbf{b}$;
    \item $\mathbf{a}\oplus d=(a_0\oplus d,a_1\oplus d,\cdots,a_{L-1}\oplus d)$, where $\oplus$ denotes the addition over $\mathbb{Z}_q$, and $d\in \mathbb{Z}_q$.
	\item $\mathbf{a}|\mathbf{b}=(a_0,a_1,\cdots,a_ {L-1},b_0,\cdots,b_{L-1})$ denotes the concatenation of $\mathbf{a}$ and $\mathbf{b}$;
\end{itemize}

\subsection{Definitions of Correlations and Sequences}
A sequence $\mathbf{a}=(a_0,a_1,\cdots,a_{L-1})$ is called $q$-ary sequence if $a_i\in\mathbb{Z}_q$ for all $i\in\{0,1,\cdots,L-1\}$. A sequence set (SS) $\mathbf{A}$ contains $N$ sequences of length $L$, denoted by $(N,L)$-SS,  is usually represented by an $N\times L$ matrix, namely $\mathbf{A} = [\mathbf{a}_0^T,\mathbf{a}_1^T,\cdots,\mathbf{a}_{N-1}^T]^T$,
where $\mathbf{a}_n~(0\leq n\leq N-1)$ is the $(n+1)$-th row sequence of $\mathbf{A}$.

\begin{Definition}
	Given two length-$L$ $q$-ary sequences $\mathbf{a}$ and $\mathbf{b}$, their aperiodic cross-correlation function (ACCF) is defined as
	\begin{align}
		R_{\mathbf{a},\mathbf{b}}(\tau)=\begin{cases}
        \sum\limits_{i=0}^{L-1-\tau}\xi^{a_i-b_{i+\tau}},&0\leq\tau\leq L-1,\\
        \sum\limits_{i=0}^{L-1+\tau}\xi^{a_{i-\tau}-b_i},&1-L\leq\tau\leq -1,\\
                0,&\mbox{otherwise.}
    \end{cases}
	\end{align}
	When $\mathbf{a}=\mathbf{b}$, $R_{\mathbf{a},\mathbf{b}}(\tau)$ is called the aperiodic auto-correlation function (AACF). For simplicity, the AACF of $\mathbf{a}$ will be written as $R_{\mathbf{a}}(\tau)$.
\end{Definition}

According to the definition of aperiodic correlation function, we have
\begin{align}\label{eq-RabRba}
  R_{\mathbf{b},\mathbf{a}}(-\tau)=R_{\mathbf{a},\mathbf{b}}^{*}(\tau),
\end{align}
where $(\cdot)^{*}$ denotes the conjugate of complex.

\begin{Definition}
	Let $\mathbf{A}$ and $\mathbf{B}$ be two $(N,L)$-SSs. The ACCF between $\mathbf{A}$ and $\mathbf{B}$ is defined as
	\begin{align}
		R_{\mathbf{A},\mathbf{B}}(\tau) = \sum\limits_{n=1}^NR_{\mathbf{a}_n,\mathbf{b}_n}(\tau),~0\leq |\tau| \leq L-1.
	\end{align}
	When $\mathbf{A}=\mathbf{B}$, $R_{\mathbf{A},\mathbf{B}}(\tau)$ is called the AACF, denoted by $R_{\mathbf{A}}(\tau)$ for short.
\end{Definition}

\begin{Definition}
  Let $\mathbf{A}$ be a $(N,L)$-SS. $\mathbf{A}$ is a called complementary sequence set (CSS) if $R_{\mathbf{A}}(\tau)=0$ for all $\tau\neq0$, denoted by $(N,L)$-CSS.
\end{Definition}

\begin{Definition}
  Let $\mathbf{A}$ be a $(N,L)$-SS. $\mathbf{A}$ is called an even-shift complementary sequence sets (ESCSS) if $R_{\mathbf{A}}(\tau)=0$ for all even number $1<|\tau|\leq L-1$, denoted by $(N,L)$-ESCSS.
\end{Definition}

\begin{Definition}
	Let $\mathcal{A} = \{\mathbf{A}_0,\mathbf{A}_1,\cdots,\mathbf{A}_{M-1}\}$ be a family of $M$ matrices, each of size $N\times L$, i.e.,
\begin{align}
  \mathbf{A}_m=\begin{bmatrix}
	\mathbf{a}_{0}^m\\
	\mathbf{a}_{1}^m\\
	\vdots\\
	\mathbf{a}_{N-1}^m
	\end{bmatrix}_{N\times L},
\end{align}
The family $\mathcal{A}$ is called a Z-complementary code set (ZCCS), denoted by $(M,N,L,Z)$-ZCCS, if
\begin{align}
	R_{\mathbf{A}_i,\mathbf{A}_j}(\tau) = \begin{cases}
		NL,& i=j,\tau=0;\\
        0, & i=j,1\leq|\tau|\leq Z-1;\\
		0, & i\neq j,0\leq|\tau|\leq Z-1.
	\end{cases}
\end{align}
When $Z=L$, the ZCCS is called mutually orthogonal complementary sequence set (MOCSS), denoted by $(M,N,L)$-MOCSS.
\end{Definition}

\begin{Lemma} [\cite{Suehiro-1988}]
For a $(M,N,L)$-MOCSS, the upper bound of family size satisfies
\begin{align}
	M\leq N.
\end{align}
When the equal sign holds, MOCSS is called a complete complementary code (CCC), denoted by $(M,L)$-CCC.
\end{Lemma}

\begin{Lemma} [\cite{Feng-2013}]
  For any $(M,N,L,Z)$-ZCCS, there exists an inequality given as
  \begin{align}\label{eq-bound-zccs}
    M\leq N\left\lfloor\frac{L}{Z}\right\rfloor.
  \end{align}
  When the equal sign holds, the ZCCS is said to be optimal.
\end{Lemma}

\subsection{Extended Boolean function}
A extended Boolean function (EBF) $f(\mathbf{x})$ is defined as a mapping $f:\mathbb{Z}_q^m\rightarrow\mathbb{Z}_q$, where $\mathbf{x}=(x_1,x_2,\cdots,x_m)$ be the $q$-ary representation of the integer $x=\sum_{k=1}^{m}x_kq^{k-1}$. Given $f(\mathbf{x})$, let $f_i=f(i_1,i_2,\cdots,i_m)$, and define the associated sequence $\mathbf{f}$ of length $q^m$ as
\begin{align}
  \mathbf{f}:=(f_0,f_1,\cdots,f_{q^m-1}).
\end{align}

\begin{Example}
  Let $m=2$ and $q=3$. The associated sequences for extended Boolean functions of $x_1,x_2,x_1x_2+1$ are
  \begin{align}\begin{split}
    \mathbf{x}_1 &:=(0,1,2,0,1,2,0,1,2),\\
    \mathbf{x}_2 &:=(0,0,0,1,1,1,2,2,2),\\
    \mathbf{x}_1\mathbf{x}_2+\mathbf{1} &:=(1,1,1,1,2,0,1,0,2),
  \end{split}\end{align}
  respectively.
\end{Example}

\section{Construction of ZCCSs Based on EBFs}
In this section, we propose a new construction of ZCCS. Before giving the main theorem, we first introduce a construction of ternary complementary sequence set.
\begin{Lemma}[\cite{Wang-Ma-2020}]\label{le-css}
  Let
  \begin{align}
    f(\mathbf{x}) &= \alpha\sum\limits_{k=1}^{m-1}x_{\pi(k)}x_{\pi(k+1)}+\sum_{k=1}^{m}c_kx_k^2+\sum_{k=1}^{m}d_kx_k+d_0,\\
    g(\mathbf{x}) &= f(\mathbf{x})+x_{\pi(1)},\\
    h(\mathbf{x}) &= f(\mathbf{x})+2x_{\pi(1)}.
  \end{align}
  where $\alpha\in\mathbb{Z}_3^{*}$, $c_k,d_k\in\mathbb{Z}_3$, $\pi$ is a permutation of the set $\{1,2,\cdots,m\}$, $m$ is a positive integer. Then the set $\{\mathbf{f},\mathbf{g},\mathbf{h}\}$ form a ternary CSS of length $3^m$.
\end{Lemma}

Generalizing the construction of Lemma \ref{le-css}, we obtain the following Theorem \ref{th-css}.
\begin{Theorem}\label{th-css}
  Let
  \begin{align}
    f(\mathbf{x}) &= \alpha\sum\limits_{k=1}^{m-1}x_{\pi(k)}x_{\pi(k+1)}+\sum_{l=1}^{q-1}\sum_{k=1}^{m}c_{l,k}x_k^l+c_0,
  \end{align}
  where $\alpha\in\mathbb{Z}_q^{*}$ and $\alpha$ is coprime with $q$, $c_0,c_{l,k}\in\mathbb{Z}_q$, $\pi$ is a permutation of the set $\{1,2,\cdots,m\}$, $m$ is a positive integer. Then the set $\mathbf{F}=\{\mathbf{f}_n\}_{n=0}^{q-1}$ form a $q$-ary CSS of length $q^m$, where
  \begin{align}
    f_n(\mathbf{x})=f(\mathbf{x})+nx_{\pi(1)}.
  \end{align}
\end{Theorem}

\begin{proof}
  The case $m=1$ is easily checked by hand, so assume $m\geq2$ and fix $\tau\neq0$. By the definition of AACF, we have
  \begin{align}\begin{split}
     R_{\mathbf{F}}(\tau) = \sum_{n=0}^{q-1}\sum_{i=0}^{q^m-1-\tau} \xi^{f_{n,i}-f_{n,i+\tau}}
  \end{split}\end{align}
  where $f_{n,i}$ is the $i$-th element of sequence $\mathbf{f}_n$. For a given integer $i$, set $j=i+\tau$ and let $(i_1,i_2,\cdots,i_m)$ and $(j_1,j_2,\cdots,j_m)$ be the $q$-ary representation of $i$ and $j$, respectively.

  \emph{Case 1:} $i_{\pi(1)}\neq j_{\pi(1)}$, then
  \begin{align}\begin{split}
     R_{\mathbf{F}}(\tau) &= \sum_{i=0}^{q^m-1-\tau}\left[\xi^{f_i-f_{j}}\left(\sum_{n=0}^{q-1}\xi^{n(i_{\pi(1)}-j_{\pi(1)})}\right)\right]\\
     &=\sum_{i=0}^{q^m-1-\tau}\left(\xi^{f_i-f_{j}} \times \frac{1-(\xi^{i_{\pi(1)}-j_{\pi(1)}})^q}{1-\xi^{i_{\pi(1)}-j_{\pi(1)}}} \right)\\
     &= 0.
  \end{split}\end{align}
  where $f_i$ is the $i$-th element of sequence $\mathbf{f}$.

  \emph{Case 2:} $i_{\pi(1)} = j_{\pi(1)}$. Since $i\neq j$, we can define $v$ to be the smallest integer for $i_{\pi(v)}\neq j_{\pi(v)}$. Let $i^{(n)}$ and $j^{(n)}$ be integers which $q$-ary  representation are different from $i$ and $j$ in only one position $\pi(v-1)$, i.e., $i_{\pi(v-1)}^{(n)}=n\oplus i_{\pi(v-1)}$ and $j_{\pi(v-1)}^{(n)}=n\oplus j_{\pi(v-1)}$ for $n=1,2,\cdots,q-1$. Then
  \begin{align}\begin{split}
    &~~~~\xi^{f_i-f_j}+\xi^{f_{i^{(1)}}-f_{j^{(1)}}}+\cdots+\xi^{f_{i^{(q-1)}}-f_{j^{(q-1)}}}\\
    &=\xi^{f_i-f_j}\Big(1+\xi^{f_{i^{(1)}}-f_{j^{(1)}}-f_i+f_j}\\
    &\hspace{8em}+\cdots+\xi^{f_{i^{(q-1)}}-f_{j^{(q-1)}}-f_i+f_j}\Big)\\
    &=\xi^{f_i-f_j}\Big(1+\xi^{\alpha(i_{\pi(v)}-j_{\pi(v)})}\\
    &\hspace{8em}+\cdots+\xi^{(q-1)\alpha(i_{\pi(v)}-j_{\pi(v)})}\Big)\\
    &=\xi^{f_i-f_j}\times\frac{1-\xi^{q\alpha(i_{\pi(v)}-j_{\pi(v)})}}{1-\xi^{\alpha(i_{\pi(v)}-j_{\pi(v)})}}\\
    &=0.
  \end{split}\end{align}
  That is to say $\xi^{f_{n,i}-f_{n,j}}+\xi^{f_{n,i^{(1)}}-f_{n,j^{(1)}}}+\cdots+\xi^{f_{n,i^{(q-1)}}-f_{n,j^{(q-1)}}}=0$ for all $n=0,1,2,\cdots,q-1$.

  Combining both cases, we see that $R_{\mathbf{F}}(\tau)=0$. Therefore, sequence set $\mathbf{F}$ is a $q$-ary CSS.
\end{proof}

In the following, we will give an example to illustrate Theorem \ref{th-css}.
\begin{Example}
  Let $q=5$, $m=2$, $\alpha=1$, $\pi(1)=2$, $\pi(2)=1$, $c_0=0$, $(c_{1,1},c_{1,2})=(1,3)$, $(c_{2,1},c_{2,2})=(2,4)$, $(c_{3,1},c_{3,2})=(1,1)$, $(c_{4,1},c_{4,2})=(0,3)$. Then we can get a CSS with size $5$ and of length $25$ based on Theorem \ref{th-css}, i.e.,
  \begin{align*}
    \mathbf{f}_1 &= (0,4,3,3,0,1,1,1,2,0,3,4,0,2,1,\\
    &\hspace{10em}4,1,3,1,1,2,0,3,2,3),\\
    \mathbf{f}_2 &= (0,4,3,3,0,2,2,2,3,1,0,1,2,4,3,\\
    &\hspace{10em}2,4,1,4,4,1,4,2,1,2),\\
    \mathbf{f}_3 &= (0,4,3,3,0,3,3,3,4,2,2,3,4,1,0,\\
    &\hspace{10em}0,2,4,2,2,0,3,1,0,1),\\
    \mathbf{f}_4 &= (0,4,3,3,0,4,4,4,0,3,4,0,1,3,2,\\
    &\hspace{10em}3,0,2,0,0,4,2,0,4,0),\\
    \mathbf{f}_5 &= (0,4,3,3,0,0,0,0,1,4,1,2,3,0,4,\\
    &\hspace{10em}1,3,0,3,3,3,1,4,3,4).
  \end{align*}
  It is easy to verify that sequence set $\mathbf{F}=\{\mathbf{f}_n\}_{n=1}^{5}$ is a $(5,25)$-CSS.
\end{Example}

Based on Theorem \ref{th-css}, we obtain the following construction of ZCCS.
\begin{Theorem}\label{th-zccs}
  Let $q\geq2$, $m\geq 2$ and $v\leq m$ be a nonnegative integer, $\pi$ be a partition of $\{1,2,\cdots,m-v\}$. Let
  \begin{align}
    f(\mathbf{x}) &= \alpha\sum\limits_{k=1}^{m-v-1}x_{\pi(k)}x_{\pi(k+1)}+\sum_{l=1}^{q-1}\sum_{k=1}^{m}c_{l,k}x_k^l+c_0,
  \end{align}
  where $\alpha\in\mathbb{Z}_q^{*}$ and $\alpha$ is coprime with $q$, $c_0,c_{l,k}\in\mathbb{Z}_q$. Then, $\mathcal{S}=\{\mathbf{S}^0,\mathbf{S}^1,\cdots,\mathbf{S}^{q^{v+1}-1}\}$ is a $(q^{v+1},q,q^m,q^{m-v})$-ZCCS, where $\mathbf{S}^{p}=\{\mathbf{s}_0^p,\mathbf{s}_1^p,\cdots,\mathbf{s}_{q-1}^p\}$ and
  \begin{align}\begin{split}
    s_n^p(\mathbf{x}) &= f(\mathbf{x})+nx_{\pi(1)}\\
    &\hspace{1em}+\beta\left({p_1x_{\pi(m-v)}}+\sum_{k=1}^{v}p_{k+1}x_{m-v+k}\right)
  \end{split}\end{align}
  for $n=0,1,\cdots,q-1$ and $p=0,1,\cdots,q^{v+1}-1$ with $q$-ary representation $(p_1,p_2,\cdots,p_{v+1})$, where $\beta\in\mathbb{Z}_q^{*}$ and $\beta$ is coprime with $q$.
\end{Theorem}

\begin{proof}
  For any $\mathbf{S}^p$, each sequence $\mathbf{s}_n^p$ $(n=0,1,\cdots,q-1)$ is divided into $q^v$ sub-sequences, i.e.,
  \begin{align}
    \mathbf{s}_n^p = \mathbf{s}_{n,0}^p|\mathbf{s}_{n,1}^p|\cdots|\mathbf{s}_{n,q^v-1}^p
  \end{align}
  More precisely, starting from the second sub-sequence, each sub-sequence can be obtained by adding a number to the first sub-sequence, i.e.,
  \begin{align}\label{eq-21}
    \mathbf{s}_n^p = \mathbf{s}_{n,0}^p|(\mathbf{s}_{n,0}^p\oplus d_1)|\cdots|(\mathbf{s}_{n,0}^p\oplus d_{q^v-1})
  \end{align}
  where $d_i\in\mathbb{Z}_q$. According to Theorem \ref{th-css}, we know that sequence set $\{\mathbf{s}_{0,0}^p,\mathbf{s}_{1,0}^p,\cdots,\mathbf{s}_{q-1,0}^p\}$ forms a CSS. Combine Eq. (\ref{eq-21}), we can obtain that $\mathbf{S}^p$ has zero auto-correlation zone for $0<\tau<q^{m-v}$.

  For any two different sequence sets $\mathbf{S}^{r}$ and $\mathbf{S}^{t}$, like Eq. (\ref{eq-21}), we can get
  \begin{align}
    \label{eq-22}\mathbf{s}_n^{r} = \mathbf{s}_{n,0}^{r}|(\mathbf{s}_{n,0}^{r}\oplus x_1)|\cdots|(\mathbf{s}_{n,0}^{r}\oplus x_{q^v-1}),\\
    \label{eq-23}\mathbf{s}_n^{t} = \mathbf{s}_{n,0}^{t}|(\mathbf{s}_{n,0}^{t}\oplus y_1)|\cdots|(\mathbf{s}_{n,0}^{t}\oplus y_{q^v-1}),
  \end{align}
  where $x_i,~y_i\in \mathbb{Z}_q$. Let $(r_1,r_2,\cdots,r_m)$ and $(t_1,t_2,\cdots,t_m)$ be $q$-ary representation of $r$ and $t$ respectively.
  \begin{itemize}
    \item Case 1: $r_1=t_1$, then $\mathbf{s}_{n,0}^{r}=\mathbf{s}_{n,0}^{t}$. For any $0<\tau<q^{m-v}$, we have
        \begin{align*}
          R_{\mathbf{S}^{r},\mathbf{S}^{t}}(\tau) &= \sum_{n=0}^{q-1}R_{\mathbf{s}_{n}^{r},\mathbf{s}_{n}^{t}}(\tau)\\
          &= q^v\sum_{n=0}^{q-1}R_{\mathbf{s}_{n,0}^{r}}(\tau)+\sum_{k=1}^{q^v-1}\xi^{x_k-y_k}
          \sum_{n=0}^{q-1}R_{\mathbf{s}_{n,0}^{r}}(\tau)\\
          &=\left(q^v+\sum_{k=1}^{q^v-1}\xi^{x_k-y_k}\right)\sum_{n=0}^{q-1}R_{\mathbf{s}_{n,0}^{r}}(\tau)\\
          &=0.
        \end{align*}
    \item Case 2: $r_{1}\neq t_{1}$, for any $0<\tau<q^{m-v}$, we have
    \begin{align*}
      R_{\mathbf{S}^{r},\mathbf{S}^{t}}(\tau) &=\sum_{n=0}^{q-1}R_{\mathbf{s}_{n}^{r},\mathbf{s}_{n}^{t}}(\tau)\\
      &=\left(q^v+\sum_{k=1}^{q^v-1}\xi^{x_k-y_k}\right)\sum_{n=0}^{q-1}R_{\mathbf{s}_{n,0}^{r},\mathbf{s}_{n,0}^{t}}(\tau).
    \end{align*}
    Similar to the proof of Theorem \ref{th-css}, we can get $\sum_{n=0}^{q-1}R_{\mathbf{s}_{n,0}^{r},\mathbf{s}_{n,0}^{t}}(\tau)=0$, so $R_{\mathbf{S}^{r},\mathbf{S}^t}(\tau)=0$.
  \end{itemize}

  Now, it only suffices to show that $R_{\mathbf{S}^{r},\mathbf{S}^{t}}(0)=0$. Our calculations yield the following result:
  \begin{align*}
    &\hspace{1em}R_{\mathbf{S}^{r},\mathbf{S}^{t}}(0) \\
    &= \sum_{n=0}^{q-1}\sum_{i=0}^{q^m-1}\xi^{s_{n,i}^{r}-s_{n,i}^{t}}\\
    &= \sum_{n=0}^{q-1}\sum_{i=0}^{q^m-1}\xi^{\beta\left[(r_{1}-t_{1})i_{\pi(m-v)}
    +\sum_{k=1}^{v}(r_{k+1}-t_{k+1})i_{m-v+k}\right]}\\
    &=0.
  \end{align*}

  To sum up, $\mathcal{S}=\{\mathbf{S}^0,\mathbf{S}^1,\cdots,\mathbf{S}^{q^{v+1}-1}\}$ is a $(q^{v+1},q,q^m,q^{m-v})$-ZCCS. This completes the proof.
\end{proof}

\begin{Remark}
  The ZCCSs from Theorem \ref{th-zccs} are optimal since $M=q^{v+1}=q\cdot(q^m/q^{m-v})=N\cdot\lfloor L/Z\rfloor$, which achieves the theoretical bound given in Eq. $(\ref{eq-bound-zccs})$.
\end{Remark}

\begin{Remark}
  The parameters of ZCCS obtained by Theorem \ref{th-zccs} can be either odd or even, which can not be realized by GBF before. In addition, our constructed ZCCSs are reduced to $(q,q^m)$-CCCs when $v=0$, the ZCCSs can have set size larger than the flock size by taking $v>0$.
\end{Remark}

Next, we give an example to illustrate Theorem \ref{th-zccs}.

\begin{Example}\label{ex-zccs}
  Taking $q=3$, $m=3$, $v=1$, $\pi(1)=1,~\pi(2)=2$, $\alpha=2$, $\beta=1$, $c_0=0$, $(c_{1,1},c_{1,2},c_{1,3})=(1,2,1)$, $(c_{2,1},c_{2,2},c_{2,3})=(0,1,2)$. Then, in Table \ref{ta-1}, we obtain a family of sequence set $\mathcal{S}=\{\mathbf{S}^0,\mathbf{S}^1,\cdots,\mathbf{S}^8\}$ by Theorem \ref{th-zccs}, it is easy to verify that $\mathcal{S}$ is a ternary $(9,3,27,9)$-ZCCS.
\end{Example}

\begin{table*}[t]
\centering
\caption{$(9,3,27,9)$-ZCCS in Example \ref{ex-zccs}}\label{ta-1}
\begin{tabular}{|c|}
\hline
 $\mathbf{S}^0$    \\
 $\begin{bmatrix}\mathbf{\mathbf{s}_0^0}\\
 \mathbf{\mathbf{s}_1^0}\\
 \mathbf{\mathbf{s}_2^0}
 \end{bmatrix} = \begin{bmatrix}
   0,1,2,0,0,0,2,1,0,0,1,2,0,0,0,2,1,0,1,2,0,1,1,1,0,2,1\\
   0,2,1,0,1,2,2,2,2,0,2,1,0,1,2,2,2,2,1,0,2,1,2,0,0,0,0\\
   0,0,0,0,2,1,2,0,1,0,0,0,0,2,1,2,0,1,1,1,1,1,0,2,0,1,2
 \end{bmatrix}$    \\ \hline
 $\mathbf{S}^1$    \\
 $\begin{bmatrix}\mathbf{\mathbf{s}_0^1}\\
 \mathbf{\mathbf{s}_1^1}\\
 \mathbf{\mathbf{s}_2^1}
 \end{bmatrix} = \begin{bmatrix}
   0,1,2,1,1,1,1,0,2,0,1,2,1,1,1,1,0,2,1,2,0,2,2,2,2,1,0\\
   0,2,1,1,2,0,1,1,1,0,2,1,1,2,0,1,1,1,1,0,2,2,0,1,2,2,2\\
   0,0,0,1,0,2,1,2,0,0,0,0,1,0,2,1,2,0,1,1,1,2,1,0,2,0,1
 \end{bmatrix}$    \\ \hline
  $\mathbf{S}^2$    \\
 $\begin{bmatrix}\mathbf{\mathbf{s}_0^2}\\
 \mathbf{\mathbf{s}_1^2}\\
 \mathbf{\mathbf{s}_2^2}
 \end{bmatrix} = \begin{bmatrix}
   0,1,2,2,2,2,0,2,1,0,1,2,2,2,2,0,2,1,1,2,0,0,0,0,1,0,2\\
   0,2,1,2,0,1,0,0,0,0,2,1,2,0,1,0,0,0,1,0,2,0,1,2,1,1,1\\
   0,0,0,2,1,0,0,1,2,0,0,0,2,1,0,0,1,2,1,1,1,0,2,1,1,2,0
 \end{bmatrix}$    \\ \hline
   $\mathbf{S}^3$    \\
 $\begin{bmatrix}\mathbf{\mathbf{s}_0^3}\\
 \mathbf{\mathbf{s}_1^3}\\
 \mathbf{\mathbf{s}_2^3}
 \end{bmatrix} = \begin{bmatrix}
   0,1,2,0,0,0,2,1,0,1,2,0,1,1,1,0,2,1,0,1,2,0,0,0,2,1,0\\
   0,2,1,0,1,2,2,2,2,1,0,2,1,2,0,0,0,0,0,2,1,0,1,2,2,2,2\\
   0,0,0,0,2,1,2,0,1,1,1,1,1,0,2,0,1,2,0,0,0,0,2,1,2,0,1
 \end{bmatrix}$    \\ \hline
   $\mathbf{S}^4$    \\
 $\begin{bmatrix}\mathbf{\mathbf{s}_0^4}\\
 \mathbf{\mathbf{s}_1^4}\\
 \mathbf{\mathbf{s}_2^4}
 \end{bmatrix} = \begin{bmatrix}
   0,1,2,1,1,1,1,0,2,1,2,0,2,2,2,2,1,0,0,1,2,1,1,1,1,0,2\\
   0,2,1,1,2,0,1,1,1,1,0,2,2,0,1,2,2,2,0,2,1,1,2,0,1,1,1\\
   0,0,0,1,0,2,1,2,0,1,1,1,2,1,0,2,0,1,0,0,0,1,0,2,1,2,0
 \end{bmatrix}$    \\ \hline
   $\mathbf{S}^5$    \\
 $\begin{bmatrix}\mathbf{\mathbf{s}_0^5}\\
 \mathbf{\mathbf{s}_1^5}\\
 \mathbf{\mathbf{s}_2^5}
 \end{bmatrix} = \begin{bmatrix}
   0,1,2,2,2,2,0,2,1,1,2,0,0,0,0,1,0,2,0,1,2,2,2,2,0,2,1\\
   0,2,1,2,0,1,0,0,0,1,0,2,0,1,2,1,1,1,0,2,1,2,0,1,0,0,0\\
   0,0,0,2,1,0,0,1,2,1,1,1,0,2,1,1,2,0,0,0,0,2,1,0,0,1,2
 \end{bmatrix}$    \\ \hline
   $\mathbf{S}^6$    \\
 $\begin{bmatrix}\mathbf{\mathbf{s}_0^6}\\
 \mathbf{\mathbf{s}_1^6}\\
 \mathbf{\mathbf{s}_2^6}
 \end{bmatrix} = \begin{bmatrix}
   0,1,2,0,0,0,2,1,0,2,0,1,2,2,2,1,0,2,2,0,1,2,2,2,1,0,2\\
   0,2,1,0,1,2,2,2,2,2,1,0,2,0,1,1,1,1,2,1,0,2,0,1,1,1,1\\
   0,0,0,0,2,1,2,0,1,2,2,2,2,1,0,1,2,0,2,2,2,2,1,0,1,2,0
 \end{bmatrix}$    \\ \hline
   $\mathbf{S}^7$    \\
 $\begin{bmatrix}\mathbf{\mathbf{s}_0^7}\\
 \mathbf{\mathbf{s}_1^7}\\
 \mathbf{\mathbf{s}_2^7}
 \end{bmatrix} = \begin{bmatrix}
   0,1,2,1,1,1,1,0,2,2,0,1,0,0,0,0,2,1,2,0,1,0,0,0,0,2,1\\
   0,2,1,1,2,0,1,1,1,2,1,0,0,1,2,0,0,0,2,1,0,0,1,2,0,0,0\\
   0,0,0,1,0,2,1,2,0,2,2,2,0,2,1,0,1,2,2,2,2,0,2,1,0,1,2
 \end{bmatrix}$    \\ \hline
    $\mathbf{S}^8$    \\
 $\begin{bmatrix}\mathbf{\mathbf{s}_0^8}\\
 \mathbf{\mathbf{s}_1^8}\\
 \mathbf{\mathbf{s}_2^8}
 \end{bmatrix} = \begin{bmatrix}
   0,1,2,2,2,2,0,2,1,2,0,1,1,1,1,2,1,0,2,0,1,1,1,1,2,1,0\\
   0,2,1,2,0,1,0,0,0,2,1,0,1,2,0,2,2,2,2,1,0,1,2,0,2,2,2\\
   0,0,0,2,1,0,0,1,2,2,2,2,1,0,2,2,0,1,2,2,2,1,0,2,2,0,1
 \end{bmatrix}$    \\ \hline
\end{tabular}
\end{table*}

\section{Construction of CCCs and MOCSSs Based on Concatenation}
In this section, a new construction of CCCs and MOCSSs will be proposed, which can make the lengths of CCCs and MOCSSs more extensive. Before giving a new construction of CCCs, we first give the following lemma.

\begin{Lemma}[\cite{Shen-Yang-2021}]\label{le-mocss}
Let $\mathbf{P}$ and $\mathbf{Q}$ be two mutually orthogonal $(N,L_1)$-ESCSSs, i.e., $R_{\mathbf{P},\mathbf{Q}}(\tau)=0$ for all $\tau$, and let $\mathcal{C}=\{\mathbf{C}^0,\mathbf{C}^1,\dots,\mathbf{C}^{M-1}\}$ be a $(M,L_2)$-CCC, where $M$ is even, and $\mathbf{C}^i$ is given by
	\begin{equation}
		\mathbf{C}^i=\begin{bmatrix}
		\mathbf{c}^i_0\\
		\mathbf{c}^i_1\\
		\vdots\\
		\mathbf{c}^i_{M-1}
		\end{bmatrix}.
	\end{equation}
For each $0\leq k\leq M-1$, define an $MN\times L_1L_2$ matrix $\mathbf{S}_k$
\begin{equation}
  \mathbf{S}^k=\begin{bmatrix}
	\mathbf{S}^{k,0}\\
	\mathbf{S}^{k,1}\\
	\vdots\\
	\mathbf{S}^{k,N-1}
\end{bmatrix},
\end{equation}
where $\mathbf{S}^{k,t}$ is an $M\times L_1L_2$ matrix, and the $j$-th row of $\mathbf{S}^{k,t}$ is given by
\begin{align}
  \mathbf{s}_j^{k,t}=\begin{cases}
    \phi(\mathbf{p}_t, \mathbf{c}_j^{2k},\mathbf{c}_j^{2k+1}), &\text{if~} 0\leq k\leq M/2-1,\\
    \phi(\mathbf{q}_t, \mathbf{c}_j^{2k-M},\mathbf{c}_j^{2k-M+1}), &\text{if~} M/2\leq k\leq M-1,
  \end{cases}
\end{align}
where
\begin{equation}
\phi(\mathbf{d},\mathbf{a},\mathbf{b})=
(\mathbf{a}\oplus d_0)|(\mathbf{b}\oplus d_1)|(\mathbf{a}\oplus d_2)|(\mathbf{b}\oplus d_3)|\cdots
\end{equation}
Then $\mathcal{S}=\{\mathbf{S}_1,\mathbf{S}_2,\cdots,\mathbf{S}_M\}$ is a $(M,MN,L_1L_2)$-MOCSS.
\end{Lemma}

\begin{Remark}
  The condition of Lemma \ref{le-mocss} can be weakened, that is, $\mathbf{P}$, $\mathbf{Q}$ are two ESCSSs and $R_{\mathbf{P},\mathbf{Q}}(\tau)=0$ for all even $\tau$.
\end{Remark}

\begin{Theorem}\label{th-escss}
  Let $\mathcal{S}=\{\mathbf{S}^0,\mathbf{S}^1,\cdots,\mathbf{S}^{M-1}\}$ be a $(M,L)$-CCC, where $M$ is a nonnegative even number. Each of $\mathbf{S}^m$ is a $(M,L)$-CSS, i.e.,
  \begin{align}
  \mathbf{S}^m=\begin{bmatrix}
	\mathbf{s}_{0}^m\\
	\mathbf{s}_{1}^m\\
	\vdots\\
	\mathbf{s}_{M-1}^m
	\end{bmatrix}_{M\times L},
  \end{align}
Let
  \begin{align}
  \mathbf{A}^m=\begin{bmatrix}
	\mathcal{I}(\mathbf{s}_{0}^m,\mathbf{s}_{1}^m)\\
    \mathcal{I}(\mathbf{s}_{2}^m,\mathbf{s}_{3}^m)\\
	\vdots\\
	\mathcal{I}(\mathbf{s}_{M-2}^m,\mathbf{s}_{M-1}^m)
	\end{bmatrix}_{M/2 \times 2L},
  \end{align}
  where $0\leq m\leq M-1$. Then, any two $\mathbf{A}^{m_1}, \mathbf{A}^{m_2}$ ($m_1\neq m_2$), $R_{\mathbf{A}^{m_1},\mathbf{A}^{m_2}}(\tau)=0$ for all even $\tau$.
\end{Theorem}

\begin{proof}
  Let $\tau=2t$ be an even number, for any $m_1\neq m_2$, we have
  \begin{align}\begin{split}
    R_{\mathbf{A}^{m_1},\mathbf{A}^{m_2}}(\tau) &= \sum\limits_{k=0}^{M/2-1}\left[R_{\mathbf{s}_{2k}^{m_1},\mathbf{s}_{2k}^{m_2}}(t)
    +R_{\mathbf{s}_{2k+1}^{m_1},\mathbf{s}_{2k+1}^{m_2}}(t)\right]\\
    &= R_{\mathbf{S}^{m_1},\mathbf{S}^{m_2}}(t)\\
    &=0,
  \end{split}\end{align}
  where last equal sign is due to that $\mathcal{S}$ is a CCC.
\end{proof}

\begin{Remark}
  According to the result of \cite{Shen-Yang-2020}, each $\mathbf{A}^m$ in Theorem \ref{th-escss} is a $(M/2,2L)$-ESCSS.
\end{Remark}

Based on Theorem \ref{th-escss}, we extend Lemma \ref{le-mocss} to construct CCCs.
\begin{Theorem}\label{th-ccc}
  Let $\mathcal{A}=\{\mathbf{A}^0,\mathbf{A}^1,\cdots,\mathbf{A}^{M_1-1}\}$ be a family of set derived from Theorem \ref{th-escss} and $\mathcal{B}$ be a $(M_2,L_2)$-CCC, where $M_1,~M_2$ are two even numbers, $\mathbf{A}^{m_1}$ ($0\leq m_1\leq M_1-1$) is an $(M_1/2,2L_1)$-ESCSS. Suppose $\mathbf{P}=\mathbf{A}^{2k}$, $\mathbf{Q}=\mathbf{A}^{2k+1}$ and $\mathcal{C}=\mathcal{B}$, we can get $M_1/2$ $(M_2,M_1M_2/2,2L_1L_2)$-MOCSSs $\{\mathcal{S}^0,\mathcal{S}^1,\cdots,\mathcal{S}^{M_1/2-1}\}$ through Lemma \ref{le-mocss}, where $0\leq k\leq M_1/2-1$. Then
  \begin{align}
    \mathcal{S}=\bigcup\limits_{i=0}^{M_1/2-1}\mathcal{S}^i
  \end{align}
  form a $(M_1M_2/2,2L_1L_2)$-CCC.
\end{Theorem}
The proof of Theorem \ref{th-ccc} is the same as the proof of Theorem 6 of \cite{Shen-Yang-2021}. Due to the limitation of space, we will not repeat it here.

In Theorem \ref{th-ccc}, we construct $(M_1M_2/2,2L_1L_2)$-CCC based on two known $(M_1,L_1)$-CCC and $(M_2,L_2)$-CCC. We aim to construct CCC of length $L_1+L_2$, but it is very difficult, therefore, we look for construction that are not as attractive as the original aim. In the following theorem, we propose a construction of MOCSS with length $L_1+L_2$, and the ratio $M/N$ is $1/2$.

\begin{Theorem}\label{th-moscc}
  Let $\mathcal{A}$ be a $(M,L_1)$-CCC and $\mathcal{B}$ be a $(M,L_2)$-CCC. Let
  \begin{align}
    \mathbf{C}^m = \begin{bmatrix}
		\mathbf{A}^m|\mathbf{B}^m\\
        \mathbf{A}^m|-\mathbf{B}^m
		\end{bmatrix}_{2M\times(L_1+L_2)}
  \end{align}
   be an matrix of size $2M\times (L_1+L_2)$, where $-\mathbf{B}^m$ means that the additive inverse of $\mathbf{B}^m$ over $\mathbb{Z}_q$, and
  \begin{align}
    \mathbf{A}^m|\mathbf{B}^m=\begin{bmatrix}
		\mathbf{a}_0^m|\mathbf{b}_0^m\\
        \mathbf{a}_1^m|\mathbf{b}_1^m\\
        \vdots\\
        \mathbf{a}_{M-1}^m|\mathbf{b}_{M-1}^m
		\end{bmatrix}_{M\times(L_1+L_2)}.
  \end{align}
  Then $\mathcal{C}=\{\mathbf{C}^0,\mathbf{C}^1,\cdots,\mathbf{C}^{M-1}\}$ be a $(M,2M,L_1+L_2)$-MOCSS.
\end{Theorem}

\begin{proof}
  Referring to \cite{Wang-Adhikary-2020}, it is a simple observation that $\mathbf{C}^m$ is a $(2M,L_1+L_2)$-CSS for all $0\leq m\leq M-1$. After that, we just need to prove that $R_{\mathbf{C}^{m_1},\mathbf{C}^{m_2}}(\tau)=0$ for any $\tau$ and $0\leq m_1\neq m_2\leq M-1$. Let us assume that $L_1 \leq L_2$.

  When $0\leq\tau<L_1$, we have
  \begin{align}\begin{split}
    &\hspace{1em}R_{\mathbf{C}^{m_1},\mathbf{C}^{m_2}}(\tau)\\
    &= \sum\limits_{n=1}^{M-1}\Big[R_{\mathbf{a}_n^{m_1},\mathbf{a}_n^{m_2}}(\tau)
    +R_{\mathbf{b}_n^{m_1},\mathbf{b}_n^{m_2}}(\tau)\\
    &\hspace{1em}+R_{\mathbf{b}_n^{m_2},\mathbf{a}_n^{m_1}}^{*}(L_1-\tau)
    +R_{\mathbf{a}_n^{m_1},\mathbf{a}_n^{m_2}}(\tau)\\
    &\hspace{1em}+R_{-\mathbf{b}_n^{m_1},-\mathbf{b}_n^{m_2}}(\tau)
    +R_{-\mathbf{b}_n^{m_2},\mathbf{a}_n^{m_1}}^{*}(L_1-\tau)\Big]\\
    &=2\left[R_{\mathbf{A}^{m_1},\mathbf{A}^{m_2}}(\tau)
    +R_{\mathbf{B}^{m_1},\mathbf{B}^{m_2}}(\tau)\right]\\
    &=0.
  \end{split}\end{align}

  When $L_1\leq\tau<L_2$, we have
  \begin{align}\begin{split}
    &\hspace{1em}R_{\mathbf{C}^{m_1},\mathbf{C}^{m_2}}(\tau)\\
    &= \sum\limits_{n=1}^{M-1}\Big[R_{\mathbf{b}_n^{m_1},\mathbf{b}_n^{m_2}}(\tau)
    +R_{\mathbf{b}_n^{m_2},\mathbf{a}_n^{m_1}}^{*}(L_1-\tau)\\
    &\hspace{1em}+R_{-\mathbf{b}_n^{m_1},-\mathbf{b}_n^{m_2}}(\tau)
    +R_{-\mathbf{b}_n^{m_2},\mathbf{a}_n^{m_1}}^{*}(L_1-\tau)\Big]\\
    &=2R_{\mathbf{B}^{m_1},\mathbf{B}^{m_2}}(\tau)\\
    &=0.
  \end{split}\end{align}

  When $L_2\leq\tau<L_1+L_2$, we have
  \begin{align}\begin{split}
    &\hspace{1em}R_{\mathbf{C}^{m_1},\mathbf{C}^{m_2}}(\tau)\\
    &= \sum\limits_{n=1}^{M-1}\Big[R_{\mathbf{b}_n^{m_2},\mathbf{a}_n^{m_1}}^{*}(L_1-\tau)
    +R_{-\mathbf{b}_n^{m_2},\mathbf{a}_n^{m_1}}^{*}(L_1-\tau)\Big]\\
    &=0.
  \end{split}\end{align}

  In conclusion, we prove that $R_{\mathbf{C}^{m_1},\mathbf{C}^{m_2}}(\tau)=0$ for any $\tau$ and $0\leq m_1\neq m_2\leq M-1$. Therefore, $\mathcal{C}$ is a $(M,2M,L_1+L_2)$-MOCSS.
\end{proof}

\begin{Remark}
  When the ratio $M/N$ is $1/2$, compared with the constructions in \cite{Shen-Yang-2021} and \cite{Wu-Chen-2021}, Theorem \ref{th-moscc} can construct MOCSS with more lengths. For example, Theorem \ref{th-moscc} can construct binary $(2,4,11)$-MOCSS, but \cite{Shen-Yang-2021} and \cite{Wu-Chen-2021} cannot.
\end{Remark}

\section{Comparison and Conclusion}
\subsection{Comparison}
At present, constructions of ZCCSs are mainly divided into direct construction and indirect construction. The direct constructions are mainly based on generalized Boolean functions \cite{Wu-Chen-2018,Sarkar-Majhi-2019,Sarkar-Roy-2020,Wu-Sahin-2021,Sarkar-Majhi-2021}. There are many indirect construction methods, such as constructions based on ZPU matrices \cite{Das-Parampalli-2020}, constructions based on GCPs, ZCPs and unitary matrices \cite{Li-Xu-2015,Adhikary-Majh-2019} etc. In Table \ref{ta-zccs}, we compare the parameters of our proposed ZCCSs with that of the previous works. Compared with the previous constructions, our proposed construction has new parameters and is a direct construction.
\begin{table*}[t]
\centering
\caption{Summary of Existing ZCCS}\label{ta-zccs}
\begin{tabular}{|c|c|c|c|c|c|c}
\hline
Source &  Based on & Parameters & Conditions &Optimality & Remark\\ \hline
\cite{Wu-Chen-2018} & GBF & $(2^{k+v},2^k,2^m,2^{m-v})$ & $v\leq m,k\leq m-v$ & Optimal & Direct\\ \hline
\cite{Sarkar-Majhi-2019} & GBF & $(2^{k+p+1},2^{k+1},2^m,2^{m-p})$ & $k+p\leq m$ & Optimal & Direct\\ \hline
\cite{Sarkar-Majhi-2021} & GBF & $(2^{n+p},2^{n},2^m,2^{m-p})$ & $p\leq m$ & Optimal & Direct\\ \hline
\cite{Sarkar-Roy-2020} & GBF & \makecell{$(2^{n},2^{n},2^{m-1}+2,$\\$2^{m-2}+2^{\pi(m-3)}+1)$}  & $m\geq3$ & Optimal & Direct\\ \hline
\cite{Wu-Sahin-2021} & GBF & $(2^{k+v},2^k,2^m,2^{m-v})$  & $v\leq m,k\leq m-v$ & Optimal & Direct\\ \hline
\cite{Wu-Sahin-2021} & GBF & $(2^{\hat{k}},2^k,2^{m-1}+2^h,Z)$ & \makecell{$v,m_1$ are two non-negative \\ integer, $m_1\leq m$, $k\leq m-v$,\\$\hat{k}\leq k,m_1\leq h\leq m-v$} &Not Optimal & Direct\\ \hline
\cite{Das-Parampalli-2020} & \makecell{Butson-type \\Hadamard Matrices} & $(K,M,K,M)$ & $K,M\geq2$ & Optimal & Indirect \\ \hline
\cite{Das-Parampalli-2020} & \makecell{Optimal \\ZPU Matrices} & $(K,M,M^{N+1}P,M^{N+1})$ & $K,M\geq2$ & Optimal & Indirect\\ \hline
\cite{Li-Xu-2015} & GCP & $(rZ,L,rs,s)$ & $r,s\geq 2,s|Z$ & Not Optimal & Indirect\\ \hline
\cite{Adhikary-Majh-2019} & ZCP & $(2^m,2^m,L,Z)$ & $Z\geq \lceil\frac{L}{2}\rceil$ & Optimal & Indirect\\ \hline
This Paper & EBF & $(q^{v+1},q,q^m,q^{m-v})$ & $v\leq m$ & Optimal & Direct\\ \hline
\end{tabular}
\end{table*}

For the comparison of MOCSS, we mainly refer to \cite{Wu-Chen-2021} and \cite{Shen-Yang-2021}, both of which propose the construction of MOCSS with $M/N$ ratio of $1/2$. The construction in \cite{Wu-Chen-2021} is based on the second order GBF, which is a direct construction, but the length of the generated sequence is relatively limited. Literature \cite{Shen-Yang-2021} constructs MOCSS with flexible lengths based on known CCCs, even-shift complementary sequence sets (ESCSSs) and concatenation operator, the rules of connection are controlled by sequences in ESCSSs. In Table \ref{ta-mocss}, we give the length $(L\leq40)$ of binary MOCSS with set size $M=2$ of flock size $N=4$ constructed by \cite{Wu-Chen-2021}, \cite{Shen-Yang-2021} and this paper, respectively. The new length obtained by our proposed construction is marked in red.
\begin{table*}[t]
\centering
\caption{Binary $(2,4)$-MOCSS for Various Lengths}\label{ta-mocss}
\begin{tabular}{|c|c|}
\hline
Source & Lengths                                                                                                                                  \\ \hline
\cite{Wu-Chen-2021}    & \makecell{$4,6,8,10,12,16,18,20,24,32,34,36,40$} \\ \hline
\cite{Shen-Yang-2021}  & \makecell{$3,4,5,6,7,8,9,10,12,14,16,17,18,$\\$20,21,22,24,26,28,32,33,34,36,40$} \\ \hline
This Paper  & \makecell{$3,4,5,6,7,8,9,10,\textcolor{red}{11},12,14,16,17,18,20,$\\$21,22,24,26,\textcolor{red}{27},28,32,33,34,36,40$} \\ \hline
\end{tabular}
\end{table*}

\subsection{Conclusion}
In this paper, based on the work in \cite{Wang-Ma-2020}, we first proposed a direct construction of $q$-ary optimal ZCCS. Compared with the previous construction, our proposed construction can fill some of the gaps in ZCCS parameters. Secondly, we generalized the result of \cite{Shen-Yang-2021} to construct CCC. Finally, we proposed a construction of MOCSS with $M/N$ ratio of $1/2$, which is very simple, but many new lengths can be constructed.

\end{document}